\documentclass[11pt]{article}

\usepackage[margin=1in]{geometry}
\usepackage{graphicx}
\usepackage[pdftex,colorlinks=true,linkcolor=blue,citecolor=blue,urlcolor=black]{hyperref}
\usepackage{amsmath, amsthm, amssymb}
\usepackage{subfigure}
\usepackage{comment}
\usepackage{url}
\usepackage{xcolor}

\newcommand{\R}{\mathbb{R}}

\newcommand{\C}{\mathbb{C}}

\newcommand{\E}{\mathbb{E}}

\newcommand{\ket}[1]{| #1 \rangle}
\newcommand{\bra}[1]{\langle #1|}
\newcommand{\proj}[1]{|#1\rangle\langle #1|}
\newcommand{\ip}[2]{\langle #1|#2 \rangle}

\newcommand{\bracket}[3]{\langle #1|#2|#3 \rangle}

\DeclareMathOperator{\tr}{tr}
\DeclareMathOperator{\diag}{diag}
\DeclareMathOperator{\rank}{rank}

\newcommand{\be}{\begin{equation}}
\newcommand{\ee}{\end{equation}}
\newcommand{\bea}{\begin{eqnarray}}
\newcommand{\eea}{\end{eqnarray}}
\newcommand{\bes}{\begin{equation*}}
\newcommand{\ees}{\end{equation*}}
\newcommand{\beas}{\begin{eqnarray*}}
\newcommand{\eeas}{\end{eqnarray*}}

\newtheorem{thm}{Theorem}
\newtheorem*{thm*}{Theorem}
\newtheorem{cor}[thm]{Corollary}

\newtheorem{lem}[thm]{Lemma}
\newtheorem*{lem*}{Lemma}

\theoremstyle{definition}
\newtheorem{dfn}{Definition}

\begin{document}

\title{Limitations on quantum dimensionality reduction}


\author{Aram W.\ Harrow$^{1,2}$, Ashley Montanaro\footnote{{\tt am994@cam.ac.uk}} $^3$ and Anthony J.\ Short$^{3}$\\ \\ {\small $^1$ Department of Computer Science \& Engineering, University of
Washington, Seattle, USA} \\
{\small $^2$ Department of Mathematics, University of Bristol, Bristol, UK
}\\  {\small $^3$ Centre for Quantum Information and Foundations, DAMTP, University of Cambridge, UK}}

\maketitle

\begin{abstract}
The Johnson-Lindenstrauss Lemma is a classic result which implies that any set of $n$ real vectors can be compressed to $O(\log n)$ dimensions while only distorting pairwise Euclidean distances by a constant factor. Here we consider potential extensions of this result to the compression of quantum states. We show that, by contrast with the classical case, there does not exist any distribution over quantum channels that significantly reduces the dimension of quantum states while preserving the 2-norm distance with high probability. We discuss two tasks for which the 2-norm distance is indeed the correct figure of merit. In the case of the trace norm, we show that the dimension of low-rank mixed states can be reduced by up to a square root, but that essentially no dimensionality reduction is possible for highly mixed states.
\end{abstract}


\section{Introduction}

The Johnson-Lindenstrauss (JL) Lemma \cite{johnson84} is a dimensionality reduction result which has found a vast array of applications in computer science and elsewhere (see e.g.\ \cite{indyk01,indyk98,kushilevitz98}). It can be stated as follows:
\begin{thm}[Johnson-Lindenstrauss Lemma \cite{johnson84}]
For all dimensions $d$, $e$, there is a distribution $\mathcal{D}$ over linear maps $\mathcal{E}:\R^d \rightarrow \R^e$ such that, for all real vectors $v$, $w$,
\[ \Pr_{\mathcal{E} \sim \mathcal{D}} [ (1-\epsilon)\|v - w \|_2 \le \| \mathcal{E}(v) - \mathcal{E}(w)\|_2 \le \|v - w \|_2] \ge 1 - \exp(-\Omega(\epsilon^2 e)), \]
\end{thm}
\noindent where $\|\cdot\|_2$ is the Euclidean ($\ell_2$) distance. The lemma is usually applied via the following corollary, which follows by taking a union bound:
\begin{cor}
\label{cor:jl}
Given a set $S$ of $n$ $d$-dimensional real vectors, there is a linear map $\mathcal{E}:\R^d \rightarrow \R^{O(\log n/\epsilon^2)}$ that preserves all Euclidean distances in $S$, up to a multiple of $1-\epsilon$. Further, there is an efficient randomised algorithm to find and implement $\mathcal{E}$.
\end{cor}
There are several remarkable aspects of this result. First, the target dimension does not depend on the source dimension $d$ at all. Second, the randomised algorithm can be simply stated as: choose a random $e$-dimensional subspace with $e = O(\log n/\epsilon^2)$, project each vector in $S$ onto this subspace, and rescale the result by a constant that does not depend on $S$. Third, this algorithm is {\em oblivious}: in other words, $\mathcal{E}$ does not depend on the vectors whose dimensionality is to be reduced.

More generally, let $\ell_p^d$ be the vector space $\R^d$ equipped with the $\ell_p$ norm $\|\cdot\|_p$. A {\em randomised embedding} from $\ell_p^d$ to $\ell_p^e$ with distortion\footnote{We use this somewhat clumsy definition of distortion for consistency with prior work.} $1/(1-\epsilon)$ and failure probability $\delta$ is a distribution $\mathcal{D}$ over maps $\mathcal{E}: \R^d \rightarrow \R^e$ such that, for all $v,w \in \R^d$,
\[ \Pr_{\mathcal{E} \sim \mathcal{D}} \left[ (1-\epsilon) \|v - w\|_p \le \|\mathcal{E}(v) - \mathcal{E}(w)\|_p \le \|v - w\|_p \right] \ge 1-\delta. \]
This definition does not allow the distance between vectors to increase; such embeddings are called contractive. The JL Lemma states that there exists a randomised embedding from $\ell_2^d$ to $\ell_2^e$ with distortion $1/(1-\epsilon)$ and failure probability $\exp(-\Omega(\epsilon^2e))$. Another natural norm to consider in this context is $\ell_1$. In this case the situation is less favourable: it has been shown by Charikar and Sahai~\cite{charikar02} that there exist $O(d)$ points in $\ell_1^d$ such that any linear embedding into $\ell_1^e$ must incur distortion $\Omega(\sqrt{d/e})$. Brinkman and Charikar later gave a set of $n$ points for which any (even non-linear) embedding achieving distortion $D$ requires $n^{\Omega(1/D^2)}$ dimensions \cite{brinkman05}.


\subsection{The JL Lemma in quantum information theory}

The JL Lemma immediately gives rise to a protocol for {\em quantum fingerprinting}~\cite{buhrman01}, or in other words efficient equality testing. Imagine that Alice and Bob each have an $n$-bit string, and are required to send quantum states of the shortest possible length to a referee, who has to use these states to determine if their bit strings are equal (this is the so-called SMP, or simultaneous message passing, model of communication complexity \cite{kushilevitz97}). Associate each bit string with an orthonormal basis vector of $\R^{2^n}$. Then the JL Lemma guarantees that there exists a map from $\R^{2^n}$ into $\R^{O(n)}$ such that the inner products between all of these $2^n$ vectors are preserved, up to a small constant. So Alice and Bob each simply apply this map to their vectors, renormalise the output (which makes very little difference to the inner products), and send the $O(\log n)$ qubit states corresponding to the resulting $O(n)$-dimensional vectors to the referee, who applies the swap test to the states \cite{buhrman01}. Given two states $\ket{\psi}$, $\ket{\phi}$, this test accepts with probability $\frac{1}{2} + \frac{1}{2}|\ip{\psi}{\phi}|^2$. As the inner products are approximately preserved by the map into $\R^{O(n)}$, the referee can distinguish between the two cases of the states he receives being equal or distinct, with constant probability.

More generally, Alice and Bob can use a similar SMP protocol to solve the following task: given quantum states $\ket{\psi_A}$, $\ket{\psi_B}$, each picked from a set of $k$ states, determine $\ip{\psi_A}{\psi_B}$ up to a constant. Whatever the initial dimension of the states, the JL Lemma (strictly speaking, an easy extension of the JL Lemma to complex vectors) guarantees that they can be compressed to $O(\log k)$ dimensions with at most constant distortion, implying that the referee can estimate $\ip{\psi_A}{\psi_B}$ up to a constant using only $O(\log \log k)$ qubits of communication.

However, there is a problem with this protocol. While it is oblivious in the sense that it does not depend on the $k$ states which are given as input, it is not oblivious in the following quantum sense: Alice and Bob each need to know what their states are in order to apply the embedding\footnote{On the other hand, if the unphysical operation of postselection is allowed, the JL Lemma can be applied directly.}. One would expect the right quantum analogue of a randomised embedding to map quantum states to quantum states in an oblivious fashion. Such an algorithm can be expressed as a distribution over quantum channels (completely positive, trace preserving (CPTP) maps \cite{nielsen00,watrous08a}), which are the class of physically implementable operations in quantum theory.

Let $\mathcal{B}(d)$ denote the set of $d$-dimensional Hermitian operators. The distance between quantum states $\rho$, $\sigma \in \mathcal{B}(d)$ can be measured using the Schatten $p$-norm $\|\rho-\sigma\|_p$, which is defined as $\|X\|_p = \left(\sum_i |\lambda_i(X)|^p\right)^{1/p}$, where $\lambda_i(X)$ is the $i$'th eigenvalue of $X$. The case $p=1$ is known as the trace norm, and $p=2$ is sometimes known as the Hilbert-Schmidt norm. We have the following definition.

\begin{dfn}
A {\em quantum embedding} from $S \subseteq \mathcal{B}(d)$ to $\mathcal{B}(e)$ in the Schatten $p$-norm, with distortion $1/(1-\epsilon)$ and failure probability $\delta$, is a distribution $\mathcal{D}$ over quantum channels $\mathcal{E}: \mathcal{B}(d) \rightarrow \mathcal{B}(e)$ such that, for all $\rho$, $\sigma \in S$,
\[ \Pr_{\mathcal{E} \sim \mathcal{D}} \left[ (1-\epsilon) \|\rho - \sigma\|_p \le \|\mathcal{E}(\rho) - \mathcal{E}(\sigma)\|_p \le \|\rho - \sigma\|_p \right] \ge 1-\delta. \]
\end{dfn}

Rather than only considering embeddings that succeed for all states in $\mathcal{B}(d)$, we generalise the definition to subsets of states. An interesting such subset is the pure states, for which one might imagine stronger embeddings can be obtained. Indeed, a closely related notion has been studied before by Winter \cite{winter04}, and more recently Hayden and Winter \cite{hayden10}, under the name of quantum identification for the identity channel. In this setting, the sender Alice has a pure state $\ket{\psi}\in\C^d$ and the receiver Bob is given the description of a pure state $\ket\phi\in\C^d$.  Alice encodes her state $\ket{\psi}$ as a quantum message using a quantum channel $\mathcal{E}:\mathcal{B}(\C^d) \rightarrow \mathcal{B}(\C^e)$ and sends it to Bob, who performs a measurement $(D_{\phi},I-D_{\phi})$ on the message. The goal is to obtain approximately the same measurement statistics as if Bob had performed the measurement $(\proj{\phi},I-\proj{\phi})$ on $\ket{\psi}$:
\[ \forall\; \ket{\psi},\ket{\phi},\; | \tr[D_{\phi}\,\mathcal{E}(\proj{\psi})] - |\ip{\psi}{\phi}|^2| \le \epsilon. \]
Winter showed in \cite{winter04} that, for constant $\epsilon$, this can be achieved with $e = O(\sqrt{d})$; note that the resulting states $\mathcal{E}(\proj{\psi})$ are highly mixed. Winter's result allows the development of a one-way protocol for testing equality of $n$-bit strings using $\frac{1}{2}\log_2 n + O(1)$ qubits of communication from Alice to Bob, which is still the best known separation between one-way quantum and classical communication complexity for total functions \cite{aaronson05a}. In our terminology, the result of \cite{winter04} shows that there exists a quantum embedding from $\mathcal{B}(d)$ to $\mathcal{B}(O(\sqrt{d}))$ that approximately preserves the trace distance between (initially) pure states. But note that one aspect of Winter's result is stronger than we need: he showed the existence of a channel such that the distance is approximately preserved between {\em all} pairs of states. Here, we are interested in finding distributions $\mathcal{D}$ over channels $\mathcal{E}$ such that, for an arbitrary pair of states, the distance is approximately preserved with high probability; this is potentially a weaker notion. In particular, it is not necessarily true that the individual channel obtained by averaging over $\mathcal{D}$ will preserve the distance between an arbitrary pair of states.

We pause to mention that the JL Lemma has found some other uses in quantum information theory. Cleve et al \cite{cleve04} used it to give an upper bound on the amount of shared entanglement required to win a particular class of nonlocal games. Gavinsky, Kempe and de Wolf \cite{gavinsky06} used it to give a simulation of arbitrary quantum communication protocols by quantum SMP protocols (with exponential overhead). Embeddings between norms have also been used. Aubrun, Szarek and Werner \cite{aubrun10,aubrun10a} have used a version of Dvoretzky's theorem on ``almost-Euclidean'' subspaces of matrices under Schatten norms to give counterexamples to the additivity conjectures of quantum information theory. And, very recently, Fawzi, Hayden and Sen \cite{fawzi10} have used ideas from the theory of low-distortion embeddings of the ``$\ell_1(\ell_2)$'' norm to prove the existence of strong entropic uncertainty relations.


\subsection{Our results}

In this paper, we show that the dimensionality reduction that can be achieved by quantum embeddings is very limited. We begin, in Section \ref{sec:2norm}, by considering the Schatten 2-norm (which is just the vector 2-norm on matrices). We show that, in stark contrast to the JL Lemma, any quantum embedding which preserves the 2-norm distance between (say) orthogonal pure states with constant distortion and constant failure probability can only achieve at most a constant reduction in dimension.

One potential criticism of this result is that the 2-norm is not usually seen as a physically meaningful distance measure, as compared with the trace norm. However, we argue in Section \ref{sec:interpret} that for certain problems the 2-norm is indeed the correct distance measure. We discuss two problems -- equality testing without a reference frame and state discrimination with a random measurement -- where the 2-norm appears naturally as the figure of merit.

In Section \ref{sec:1norm} we turn to the trace norm, for which we have upper and lower bounds. On the upper bound side, we extend the result of Winter \cite{winter04} to show that low-rank mixed states are also amenable to dimensionality reduction; roughly speaking, $d$-dimensional mixed states of rank $r$ can be embedded into $O(\sqrt{rd})$ dimensions with constant distortion. On the other hand, we show using the 2-norm lower bound that highly mixed states cannot be embedded into low dimension: there is a lower bound of $\Omega(\sqrt{d} \frac{\|\rho-\sigma\|_1}{\|\rho-\sigma\|_2})$ on the target dimension of any constant distortion trace norm embedding that succeeds with constant probability for the pairs $U \rho U^\dag$, $U \sigma U^\dag$ for all unitary operators $U$. In particular, this implies an $\Omega(\sqrt{d})$ lower bound for any embedding which succeeds for a unitarily invariant set of states. In the case that $|\rho-\sigma|$ is proportional to a projector (i.e.\ all non-zero eigenvalues of $\rho-\sigma$ are equal in absolute value), our upper and lower bounds coincide.

Finally, some notes on miscellaneous notation. $F_d$ will denote the unitary operator which swaps (or flips) two $d$-dimensional quantum systems (i.e.\ $F_d = \sum_{i,j=1}^d \ket{i}\bra{j}\otimes \ket{j}\bra{i}$), and $I_d$ will denote the $d$-dimensional identity matrix. Whenever we say that $U \in U(d)$ is a random unitary operator, we mean that $U$ is picked uniformly at random according to Haar measure on the unitary group $U(d)$.


\section{Dimensionality reduction in the 2-norm}
\label{sec:2norm}

We now show that quantum dimensionality reduction in the 2-norm is very limited.

\begin{thm}
\label{thm:nojl}
Let $\mathcal{D}$ be a distribution over quantum channels (CPTP maps) $\mathcal{E} : \mathcal{B}(\C^d) \rightarrow \mathcal{B}(\C^e)$ such that, for fixed quantum states $\rho \neq \sigma$ and for all unitary operators $U \in U(d)$,
\[ \Pr_{\mathcal{E} \sim \mathcal{D}} [ \| \mathcal{E}(U\rho U^\dag) - \mathcal{E}(U \sigma U^\dag)\|_2 \ge (1-\epsilon) \| U \rho U^\dag - U \sigma U^\dag \|_2 ] \ge 1-\delta \]
for some $0 \le \epsilon,\delta \le 1$. Then $e \ge (1-\delta)(1-\epsilon)^2d$.
\end{thm}

Note that the above lower bound on target dimension holds for any embedding of a unitarily invariant set of states. For example, taking $\rho$ and $\sigma$ to be orthogonal pure states and inserting $\epsilon=\delta=0$ recovers the (unsurprising) result that any embedding that exactly preserves distances between all orthogonal pure states with certainty must satisfy $e \ge d$. More generally, if we have an embedding which succeeds with constant probability and has constant distortion, the target dimension can be no smaller than $\Omega(d)$. In order to prove the theorem, we will need the following two technical lemmas, which are proved in Appendix \ref{sec:2normlemmas}.

\newcounter{lems}\setcounter{lems}{\value{thm}}

\begin{lem}
\label{lem:flipbound}
Let $\mathcal{E}:\mathcal{B}(\C^d) \rightarrow \mathcal{B}(\C^e)$ be a quantum channel (CPTP map). Then
\[ \tr[F_e\,\mathcal{E}^{\otimes 2}(F_d)] \le de. \]
\end{lem}

\begin{lem}
\label{lem:twirl}
Let $\rho$ and $\sigma$ be $d$-dimensional quantum states. Then
\[ \int U^{\otimes 2} (\rho - \sigma)^{\otimes 2} (U^{\dag})^{\otimes 2} dU = \frac{\|\rho-\sigma\|_2^2}{d^2-1} \left( F_d - \frac{I_{d^2}}{d} \right). \]
\end{lem}

The following lemma is the key to most of the results in this paper.

\begin{lem}
\label{lem:avg}
Let $\rho$ and $\sigma$ be quantum states and let $\mathcal{E}: \mathcal{B}(\C^d) \rightarrow \mathcal{B}(\C^e)$ be a quantum channel. Then
\[ \int \| \mathcal{E}(U\rho U^\dag) - \mathcal{E}(U\sigma U^\dag)\|_2^2\,dU \le \frac{d(e^2-1)}{e(d^2-1)} \|\rho-\sigma\|_2^2. \]
\end{lem}

\begin{proof}
We have
\beas
\int \| \mathcal{E}(U\rho U^\dag) - \mathcal{E}(U\sigma U^\dag)\|_2^2\,dU &=& \int \| \mathcal{E}(U(\rho -\sigma)U^\dag)\|_2^2\,dU\\
&=& \int\tr[F_e\,\mathcal{E}(U(\rho -\sigma)U^\dag)^{\otimes 2}]\,dU\\
&=& \tr\left[F_e\,\mathcal{E}^{\otimes 2}\left( \int U^{\otimes 2}(\rho -\sigma)^{\otimes 2}(U^\dag)^{\otimes 2}\, dU \right)\right]\\
&=& \frac{\|\rho-\sigma\|_2^2}{d^2-1} \tr\left[F_e\,\mathcal{E}^{\otimes 2} \left( F_d - \frac{I_{d^2}}{d} \right)\right]\\
&\le& \frac{\|\rho-\sigma\|_2^2}{d^2-1} \left( de - d \tr[\mathcal{E}(I_d/d)^2] \right)\\
&\le& \frac{d(e^2-1)}{e(d^2-1)} \|\rho-\sigma\|_2^2.
\eeas
We use linearity of $\mathcal{E}$ in the first equality, and the second equality is the tensor product trick $\tr[X^2] = \tr[F_e X^{\otimes 2}]$ for $e$-dimensional operators $X$. The fourth equality is Lemma \ref{lem:twirl}, the first inequality is Lemma \ref{lem:flipbound}, and the second inequality is simply $\tr \rho^2 \ge 1/e$ for all $e$-dimensional states $\rho$.
\end{proof}

We are finally ready to prove Theorem \ref{thm:nojl}.

\begin{proof}[Proof of Theorem \ref{thm:nojl}]
We will prove something slightly stronger: that for a random $U$, the 2-norm is not approximately preserved under a map $\mathcal{E}$ picked from $\mathcal{D}$, unless $e$ is almost as large as $d$. So assume
\[ \Pr_{\mathcal{E} \sim \mathcal{D},\,U \in U(d)} \left[ \| \mathcal{E}(U \rho U^\dag) - \mathcal{E}(U \sigma U^\dag)\|_2 \ge (1-\epsilon) \| U\rho U^\dag - U\sigma U^\dag \|_2 \right] \ge 1-\delta, \]
or equivalently
\[ \Pr_{\mathcal{E} \sim \mathcal{D},\,U \in U(d)} \left[ \| \mathcal{E}(U\rho U^\dag) - \mathcal{E}(U \sigma U^\dag)\|_2^2 \ge (1-\epsilon)^2 \| \rho - \sigma\|_2^2 \right] \ge 1-\delta, \]
where we use the unitary invariance of the 2-norm. By Markov's inequality, this implies that
\[ 
\int_{\mathcal{E} \sim \mathcal{D}} \int \| \mathcal{E}(U \rho U^\dag) - \mathcal{E}(U \sigma U^\dag)\|_2^2\,dU \ge (1-\delta) (1-\epsilon)^2\| \rho - \sigma\|_2^2, \]
implying in turn that there must exist some $\mathcal{E}$ such that
\[ \int \| \mathcal{E}(U \rho U^\dag) - \mathcal{E}(U \sigma U^\dag)\|_2^2\,dU \ge (1-\delta)(1-\epsilon)^2 \| \rho - \sigma\|_2^2. \]
So let $\mathcal{E}: \mathcal{B}(\C^d) \rightarrow \mathcal{B}(\C^e)$ be a quantum channel that does satisfy this inequality. Then we have
\beas
(1-\delta)(1-\epsilon)^2 \| \rho - \sigma\|_2^2 &\le& \int \| \mathcal{E}(U\rho U^\dag) - \mathcal{E}(U \sigma U^\dag) \|_2^2\,dU
\le \left( \frac{e}{d}\right) \|\rho - \sigma\|_2^2,
\eeas
where the second inequality follows from Lemma \ref{lem:avg}, assuming that $e \le d$. We have shown that $e \ge (1-\delta)(1-\epsilon)^2d$, completing the proof of the theorem.
\end{proof}


\section{Operational meaning of the 2-norm}
\label{sec:interpret}

In this section, we discuss the meaning of the 2-norm distance between quantum states. It is usually assumed that the trace norm is the ``right'' measure of distance between states, and proofs going via the 2-norm usually do so only for calculational simplicity. However, here we argue that the 2-norm is of interest in its own right, by giving two operational interpretations of this distance measure.


\subsection{Equality testing without a reference frame}

Consider the following equality-testing game. We are given a description of two different states $\rho$ and $\sigma$. An adversary prepares two systems in one of the states $\rho \otimes \rho$, $\sigma \otimes \sigma$, $\rho \otimes \sigma$ or $\sigma \otimes \rho$, with equal probability of each. He then applies an unknown unitary $U$ to each system (i.e.\ he applies $U \otimes U$ to the joint state). Our task is to determine whether the two systems have the same state or different states. This models equality testing in a two-party scenario in which the preparer and tester do not share a reference frame \cite{bartlett03}. One protocol for solving this task is simply to apply the swap test \cite{buhrman01} to the two states we are given, output ``same'' if the test accepts, and ``different'' otherwise. When applied to two states $\rho$, $\sigma$ this test accepts with probability $\frac{1}{2} + \frac{1}{2} \tr \rho\,\sigma$, so for any $U$ the overall probability of success is
\[ \frac{1}{4} \left( \frac{1}{2} + \frac{1}{2} \tr[\rho^2] \right) + \frac{1}{4} \left( \frac{1}{2} + \frac{1}{2} \tr [\sigma^2] \right) + \frac{1}{2}\left( \frac{1}{2} - \frac{1}{2} \tr[ \rho\,\sigma ] \right) = \frac{1}{2} + \frac{1}{8} \| \rho - \sigma \|_2^2. \]
Using our previous result, we now show that this is optimal.
\begin{thm}
The maximal probability of success of the above game is $\frac{1}{2} + \frac{1}{8} \| \rho - \sigma \|_2^2$.
\end{thm}

\begin{proof}
Let $(M,I-M)$ be an arbitrary POVM where the operator $M$ corresponds to the answer ``same''. Then the probability of success achieved by this POVM for a given $U$ is $\frac{1}{2} + \frac{1}{2}B$, where $B$ is the {\em bias}, which is equal to
\[ \tr\left[ M\left( \frac{1}{2}(U \rho U^\dag \otimes U \rho U^\dag + U \sigma U^\dag \otimes U \sigma U^\dag) - \frac{1}{2}(U \rho U^\dag \otimes U \sigma U^\dag + U \sigma U^\dag \otimes U \rho U^\dag) \right)\right]. \]
If the adversary adopts the strategy of picking $U$ uniformly at random, the average bias obtained is
\[ \frac{1}{2} \tr\left[ M \int U^{\otimes 2}(\rho \otimes \rho + \sigma \otimes \sigma - \rho \otimes \sigma - \sigma \otimes \rho)(U^{\dag})^{\otimes 2} dU \right] = \frac{1}{2} \tr\left[ M \int U^{\otimes 2}(\rho - \sigma)^{\otimes 2} (U^\dag)^{\otimes 2}\right], \]
which by Lemma \ref{lem:twirl} is equal to
\[ \frac{\|\rho-\sigma\|_2^2}{2(d^2-1)} \tr \left[M\left( F_d - \frac{I_{d^2}}{d} \right)\right]. \]
This expression is maximised by setting $M$ equal to a projector onto the subspace spanned by the eigenvectors of $F_d - \frac{I_{d^2}}{d}$ with positive eigenvalues. As $F_d$ has $d(d+1)/2$ eigenvalues equal to 1, and $d(d-1)/2$ eigenvalues equal to $-1$, we obtain $\tr \left[M\left( F_d - \frac{I_{d^2}}{d} \right)\right] = (d^2-1)/2$. This implies that the average bias is at most $\frac{1}{4} \| \rho - \sigma \|_2^2$. As the worst-case bias can only be lower, this implies the claimed result.
\end{proof}


\subsection{Performing a random measurement}

The second game we will discuss is state discrimination with a fixed or random measurement. Imagine we are given a state which is promised to be either $\rho$ or $\sigma$, with equal probability of each, and we wish to determine which is the case. It is well-known that the largest bias achievable by choosing an appropriate measurement is $\frac{1}{2} \|\rho-\sigma\|_1$ (recall from the previous section that the bias $B$ and the success probability $p$ have the relationship $p = \frac{1}{2} + \frac{B}{2}$). But how well can we do if the measurement we apply does not in fact depend on $\rho$ and $\sigma$?

We will see that $\|\rho - \sigma\|_2$ is closely related to the optimal bias achievable by performing one of the following two measurements, and deciding whether the state is $\rho$ or $\sigma$ based on the outcome.
\begin{itemize}
\item The uniform (isotropic) POVM whose measurement elements consist of normalised projectors onto all states $\ket{\psi}$;
\item A projective measurement in a random basis (i.e.\ applying a random unitary operator and measuring in the computational basis).
\end{itemize}
In general, the largest bias achievable by measuring a POVM $M$ which consists of measurement operators $M_i$ can be written as
\[ \frac{1}{2} \sum_i |\tr[M_i (\rho-\sigma)]|. \]
Each measurement operator of the uniform POVM is given by the projector onto some state $\ket{\psi}$, normalised by a factor of $d$ (to check that this is right, note that
\[ d \int d\psi \proj{\psi} = d \left( \frac{I_d}{d} \right) = I_d \]
as expected). So the bias induced by the uniform POVM is
\[ \frac{d}{2} \int d\psi |\bracket{\psi}{(\rho-\sigma)}{\psi}|. \]
In the case of a measurement in a random basis $U \in U(d)$, we can calculate the {\em expected} bias as follows:
\beas
\frac{1}{2} \E_U \sum_{i=1}^d |\bracket{i}{U^{\dag}(\rho-\sigma)U}{i}| &=& \frac{1}{2} \sum_{i=1}^d \E_U |\bracket{i}{U^{\dag}(\rho-\sigma)U}{i}| = \frac{1}{2} \sum_{i=1}^d \E_U |\bracket{1}{U^{\dag}(\rho-\sigma)U}{1}|\\
&=& \frac{d}{2} \int d\psi |\bracket{\psi}{(\rho-\sigma)}{\psi}|;
\eeas
so these quantities are the same. They are also closely related to the 2-norm distance, as we will now see.

\newcounter{thm2}\setcounter{thm2}{\value{thm}}

\begin{thm}
\label{thm:2norm}
Let $\rho$, $\sigma$ be $d$-dimensional quantum states. Then
\[ \frac{1}{3} \|\rho-\sigma\|_2 \le d \int d\psi |\bracket{\psi}{(\rho-\sigma)}{\psi}| \le \|\rho-\sigma\|_2. \]
\end{thm}

The lower bound in Theorem \ref{thm:2norm} was shown by Ambainis and Emerson~\cite{ambainis07b} (see also the proof of Matthews, Wehner and Winter \cite{matthews09a}), and the upper bound is not hard. However, as this result does not appear to be widely known, we include a proof (which is essentially the same as that of~\cite{matthews09a}) in Appendix \ref{sec:isoproof}.

In fact, the corresponding upper and lower bounds on the bias hold for any fixed POVM whose measurement vectors form a 4-design~\cite{ambainis07b}, and the upper bound even holds for any fixed POVM whose vectors form a 2-design. This result can be useful in cases where one wishes to perform state discrimination without necessarily being able to construct the optimal measurement efficiently \cite{sen06}. See the work \cite{matthews09a} for much more detail on the bias achievable in state discrimination with fixed measurements.


\section{Dimensionality reduction in the trace norm}
\label{sec:1norm}

In this section we consider embeddings that reduce dimension while preserving the trace norm distance between states. As no quantum channel can increase this distance, we first observe that any such embedding will automatically be contractive.


\subsection{Upper bound}
\label{sec:trupper}

It was previously shown by Winter \cite{winter04} that, in our language, $d$-dimensional pure states can be embedded into $\mathcal{B}(O(\sqrt{d}))$ with constant distortion. We now extend this result to general mixed states, by showing that rank $r$ mixed states can be embedded into dimension $O(\sqrt{rd})$ with constant distortion.

The embedding is conceptually very simple: apply a random unitary and trace out a subsystem. However, when the target dimension $e$ does not divide $d$, we are forced to consider random isometries $V:\C^d \rightarrow \C^e \otimes \C^{\lceil d/e \rceil}$ instead of unitaries, where $\lceil x \rceil$ is the smallest integer $y$ such that $y \ge x$. Recall that an isometry is a norm-preserving linear map, i.e.\ a map taking an orthonormal basis of one space to an orthonormal set of vectors in another (potentially larger) space. A random isometry is defined as a fixed isometry followed by a random unitary.

Formally, our embedding is a distribution over the following quantum channels $\mathcal{E}_V$.

\begin{dfn}
\label{dfn:channel}
Let $d$ and $e$ be positive integers such that $e \le d$. For any isometry $V:\C^d \rightarrow \C^e \otimes \C^{\lceil d/e \rceil}$, let $\mathcal{E}_V : \mathcal{B}(\C^d) \rightarrow \mathcal{B}(\C^e)$ be the quantum channel that consists of performing $V$, then tracing out (discarding) the second subsystem.
\end{dfn}

We now analyse the performance of the embedding obtained by picking a random $V$ and applying this channel.

\begin{thm}
\label{thm:trupper}
Let $d$ be a positive integer, and let $\rho$ and $\sigma$ be arbitrary $d$-dimensional mixed states such that $\rho$ has rank $r$. Fix $\epsilon$ such that $0 < \epsilon < 1$. For any $e$ such that $2\sqrt{rd/\epsilon} \le e \le d$, let $\mathcal{D}$ be the distribution on channels $\mathcal{E}_V : \mathcal{B}(\C^d) \rightarrow \mathcal{B}(\C^e)$ that is uniform on isometries $V:\C^d \rightarrow \C^e \otimes \C^{\lceil d/e \rceil}$. Then
\[ \Pr_{\mathcal{E}_V \sim \mathcal{D}} [ \| \mathcal{E}_V(\rho) - \mathcal{E}_V(\sigma)\|_1 \ge (1-\epsilon)\| \rho - \sigma \|_1] \ge 1 - d\,\exp(-K\epsilon d),
\]
for a universal constant $K$ which may be taken to be $(1-\ln 2)/(2 \ln 2) \approx 0.22$.
\end{thm}

In order to prove this theorem, we will need the following technical lemma, which is proven in Appendix \ref{sec:projsupp}.

\newcounter{thm3}\setcounter{thm3}{\value{thm}}

\begin{lem}
\label{lem:projsupp}
Let $\mathcal{H} = \mathcal{H}_A \otimes \mathcal{H}_B$ be a finite-dimensional Hilbert space decomposed into subsystems $A$ and $B$. For any projector $P$ onto a subspace of $\mathcal{H}$, let $P^\perp = I-P$ be the projector onto the orthogonal subspace, and let $D$ be the projector onto the support of $\tr_B P$. Then, for any $\ket{\psi} \in \mathcal{H}$,
\[ \tr[ (D \otimes I)P^\perp \proj{\psi} P^\perp] \le \tr[(D \otimes I) \proj{\psi}] \tr[P^\perp \proj{\psi}]. \]
\end{lem}

We will also need the following useful result of Bennett et al \cite{bennett05} (see also~\cite{winter04}).

\begin{lem}
\label{lem:projconc}
Let $\ket{\psi}$ be a $d$-dimensional pure state, let $P$ be the projector onto a $t$-dimensional subspace of $\C^d$, and let $U \in U(d)$ be picked according to Haar measure. Then, for any $\delta\ge0$,
\[ \Pr_U \left[\tr[U P U^\dag \proj{\psi}] \ge (1+\delta) \frac{t}{d} \right] \le \exp(-t(\delta - \ln(1+\delta))/(\ln 2)). \]
\end{lem}

\begin{proof}[Proof of Theorem \ref{thm:trupper}]
We will upper bound the probability of the embedding failing, i.e.\
\[ \Pr_V [ \| \mathcal{E}_V (\rho - \sigma) \|_1 < (1-\epsilon)\| \rho - \sigma \|_1 ]. \]
Let $S^+$, $S^-$ be the disjoint sets of indices of $(\rho - \sigma)$'s positive and negative eigenvalues, respectively. Set $s=|S^+|$, and note that $s\le \rank(\rho) = r$ \cite[Corollary III.2.3]{bhatia97}. For a fixed $V$, expand $V(\rho - \sigma)V^\dag$ as follows:
\[ V(\rho - \sigma)V^\dag = \sum_{i\in S^+} \lambda_i \proj{\psi_i} - \sum_{i\in S^-} \mu_i \proj{\psi_i} \]
for some orthonormal vectors $\ket{\psi_i} \in \C^e \otimes \C^{\lceil d/e \rceil}$ and positive coefficients $\lambda_i$, $\mu_i$. Note that
\[ \sum_{i \in S^+} \lambda_i = \sum_{i\in S^-} \mu_i = \|\rho-\sigma\|_1/2. \]
For any states $\rho'$ and $\sigma'$, it holds that
\[ \|\rho' - \sigma'\|_1 = 2 \sup_{0 \le M \le I} \tr M(\rho' - \sigma'); \]
in a protocol for distinguishing $\rho'$ and $\sigma'$, $M$ is a measurement operator corresponding to the outcome that the state was $\rho'$. Thus, in order for it to hold that $\| \mathcal{E}_V(\rho - \sigma) \|_1 \ge (1 - \epsilon) \| \rho - \sigma \|_1$, it suffices to exhibit an operator $M$ such that $0 \le M \le I$ and
\[ \tr[M(\mathcal{E}_V(\rho - \sigma)) ] \ge (1 - \epsilon)\| \rho - \sigma \|_1/2 = (1 - \epsilon) \sum_{i \in S^+} \lambda_i. \]
To find such an operator, set
\[ P_V := \sum_{i \in S^+} \proj{\psi_i}. \]
Note that $P_V$ is the projector onto a random $s$-dimensional subspace of $\C^e \otimes \C^{\lceil d/e \rceil}$. Now let $D_V$ be the projector onto the support of $\tr_B P_V$. Then
\be
\label{eqn:union}
\tr[D_V \mathcal{E}_V (\rho - \sigma)] = \sum_{i\in S^+} \lambda_i \tr[D_V \tr_B \proj{\psi_i}] - \sum_{i\in S^-} \mu_i \tr [D_V \tr_B \proj{\psi_i}].
\ee
For all $i \in S^+$, $\tr[D_V \tr_B \proj{\psi_i}] = 1$, and for all $i \in S^-$, it holds that $\tr[P_V \proj{\psi_i}] = 0$. Aside from this constraint, each individual state $\ket{\psi_i}$, $i \in S^-$, is picked at random and can be expressed in terms of a general random state $\ket{\eta} \in \C^e \otimes \C^{\lceil d/e \rceil}$ as
\[ \ket{\psi_i} = \frac{P_V^\perp \ket{\eta}}{\|P_V^\perp \ket{\eta} \|_2}, \]
where $P_V^\perp = I-P_V$ and the denominator is non-zero with probability 1. Then
\[
\tr[(D_V \otimes I) \proj{\psi_i}] = \frac{\tr [(D_V \otimes I)(P_V^\perp\proj{\eta}P_V^\perp)] }{\tr[P_V^\perp \proj{\eta}]}\le \tr[(D_V \otimes I) \proj{\eta}],
\]
where the inequality is Lemma \ref{lem:projsupp}. For any $e$ such that $e \ge s \lceil d/e \rceil$, $D_V$ has rank $s \lceil d/e \rceil$ with probability 1. So, for any such $e$, $D_V \otimes I$ has rank $s \lceil d/e \rceil^2$ with probability 1. Applying Lemma \ref{lem:projconc}, for any $\delta \ge 0$,
\[ \Pr_{\ket{\eta}} \left[\tr [(D_V \otimes I)\proj{\eta}] \ge (1 + \delta)\frac{s \lceil d/e \rceil^2}{e \lceil d/e \rceil}\right ] \le \exp(-s\lceil d/e \rceil^2(\delta - \ln (1+\delta))/(\ln 2)) \]
and hence
\[ \Pr_V \left[\tr [(D_V \otimes I)\proj{\psi_i}] \ge (1 + \delta)\frac{s \lceil d/e \rceil}{e}\right] \le \exp(-s\lceil d/e \rceil^2(\delta - \ln (1+\delta))/(\ln 2)). \]
Using a union bound over $S^-$ in eqn.\ (\ref{eqn:union}), for any $e$ satisfying $e \ge s \lceil d/e \rceil$ it holds that
\[ \Pr_V \left[\tr[D_V \mathcal{E}_V (\rho - \sigma)] \le \sum_{i\in S^+} \lambda_i - (1 + \delta)\frac{s \lceil d/e \rceil}{e} \sum_{i\in S^-} \mu_i\right] \le d\,\exp(-s\lceil d/e \rceil^2(\delta - \ln (1+\delta))/(\ln 2)). \]
We now set $\delta = \frac{\epsilon e}{s \lceil d/e \rceil} - 1$. This gives the following bound, valid when $\epsilon e \ge s\lceil d/e \rceil$:
\beas
\Pr_V \left[\tr[D_V \mathcal{E}_V (\rho - \sigma)] \le (1 - \epsilon)\| \rho - \sigma \|_1/2 \right] &\le& d\,\exp\left(-s\lceil d/e \rceil^2\left(\frac{\epsilon e}{s \lceil d/e \rceil} - 1 - \ln \left(\frac{\epsilon e}{s \lceil d/e \rceil}\right)\right)/(\ln 2)\right)\\
&\le& d\,\exp\left(-s(d/e)\lceil d/e \rceil\left(\frac{\epsilon e}{s \lceil d/e \rceil} - 1 - \ln \left(\frac{\epsilon e}{s \lceil d/e \rceil}\right)\right)/(\ln 2)\right)\\
&=& d\,\exp\left(-\epsilon d\left(1 - \frac{s\lceil d/e \rceil}{\epsilon e} \left(1 + \ln \left(\frac{\epsilon e}{s \lceil d/e \rceil}\right)\right)\right)/(\ln 2)\right).
\eeas
Now the function $f(x) = x(1+ \ln(1/x))$ increases with $x$ in the range $0 < x \le 1$, so for any $e$ such that $\frac{s \lceil d/e \rceil}{\epsilon e} \le 1/2$, we have
\beas
\Pr_V \left[\tr[D_V \mathcal{E}_V (\rho - \sigma)] \le (1 - \epsilon)\| \rho - \sigma \|_1/2 \right] &\le& d\,\exp(-\epsilon d(1 - f(1/2))/(\ln 2))\\
&=& d\,\exp(-\epsilon d(1 - \ln 2)/(2 \ln 2)).
\eeas
Thus this inequality holds for any $e$ such that $\epsilon e \ge 2s \lceil d/e \rceil$. As $\lceil d/e \rceil \le 2d/e$ for $e \le d$, this will be satisfied for any $e \ge 2 \sqrt{sd/\epsilon}$, and in particular any $e \ge 2 \sqrt{rd/\epsilon}$, implying for any such $e$
\[ \Pr_{\mathcal{E}_V \sim \mathcal{D}} [ \| \mathcal{E}_V(\rho) - \mathcal{E}_V(\sigma)\|_1 \le (1-\epsilon)\| \rho - \sigma \|_1] \le d\,\exp(-\epsilon d(1-\ln 2)/(2 \ln 2)) \]
as required.
\end{proof}

Although this result is expressed in terms of the rank of the input states, a similar result would apply to states which are very close (in trace norm) to having low rank, but for simplicity we do not discuss this here.


\subsection{Lower bound}

It turns out that Lemma \ref{lem:avg} is also strong enough to give a bound on embeddings of the trace norm, via a similar proof to that of Theorem \ref{thm:nojl}. Charikar and Sahai \cite{charikar02} showed that there exist a set of $O(d)$ $d$-dimensional vectors whose dimension cannot be significantly reduced while preserving their $\ell_1$ distances. One might expect the same to be true for the trace norm, as the trace norm on diagonal matrices is just the $\ell_1$ norm of the diagonal entries. However, note that this does not follow immediately from Charikar and Sahai's work, as it is conceivable that an embedding mapping diagonal to non-diagonal matrices could do better. Nevertheless, we now show that dimensionality reduction is impossible for some sets of highly mixed states.

\begin{thm}
\label{thm:embedtrace}
Let $\mathcal{D}$ be a distribution over quantum channels (CPTP maps) $\mathcal{E} : \mathcal{B}(\C^d) \rightarrow \mathcal{B}(\C^e)$ such that, for fixed quantum states $\rho\neq\sigma$ and for all unitary $U$,
\[ \Pr_{\mathcal{E} \sim \mathcal{D}} [ \| \mathcal{E}(U \rho U^\dag) - \mathcal{E}(U \sigma U^\dag)\|_1 \ge (1-\epsilon) \| U\rho U^\dag - U\sigma U^\dag \|_1 ] \ge 1-\delta \]
for some $0 \le \epsilon,\delta \le 1$. Then
\[ e \ge (1-\delta)(1-\epsilon) \sqrt{d} \frac{\|\rho-\sigma\|_1}{\|\rho-\sigma\|_2}. \]
In particular, if $\rho$ and $\sigma$ are orthogonal pure states, then $e \ge (1-\delta)(1-\epsilon) \sqrt{2d}$, and if $\rho$ and $\sigma$ are proportional to projectors onto orthogonal $d/2$-dimensional subspaces, $e \ge (1-\delta)(1-\epsilon) d$.
\end{thm}

So we see that achieving any significant dimensionality reduction for arbitrary highly mixed states is impossible, and even for pure states the dimension can only be reduced by a square root (which was already known \cite{winter04}).

\begin{proof}
For a randomly chosen $U$, we have
\[ \Pr_{\mathcal{E} \sim \mathcal{D},\,U \in U(d)} [ \| \mathcal{E}(U \rho U^\dag) - \mathcal{E}(U \sigma U^\dag)\|_1 \ge (1-\epsilon) \| U\rho U^\dag - U\sigma U^\dag \|_1 ]\,dU \ge 1-\delta, \]
and use Markov's inequality and the unitary invariance of the trace norm to obtain
\[  \int_{\mathcal{E} \sim \mathcal{D}} \int \| \mathcal{E}(U \rho U^\dag) - \mathcal{E}(U \sigma U^\dag) \|_1\,dU \ge (1 - \delta) (1-\epsilon) \| \rho - \sigma \|_1. \]
Thus there must exist some $\mathcal{E}$ such that
\[ \int \| \mathcal{E}(U \rho U^\dag) - \mathcal{E}(U \sigma U^\dag) \|_1\,dU \ge (1 - \delta) (1-\epsilon) \| \rho - \sigma \|_1. \]
Simply estimating the 1-norm by the 2-norm and using Jensen's inequality, we get the bounds
\beas
(1 - \delta) (1-\epsilon) \| \rho - \sigma \|_1 &\le& \sqrt{e} \int \| \mathcal{E}(U \rho U^\dag) - \mathcal{E}(U \sigma U^\dag) \|_2\,dU\\
&\le& \sqrt{e} \left( \int \| \mathcal{E}(U \rho U^\dag) - \mathcal{E}(U \sigma U^\dag) \|_2^2\,dU \right)^{1/2}\\
&\le& \left( \frac{e}{\sqrt{d}} \right) \|\rho-\sigma\|_2,
\eeas
where the last inequality follows from Lemma \ref{lem:avg}, assuming that $e \le d$. Rearranging gives the theorem.
\end{proof}

This implies that the protocol of Theorem \ref{thm:trupper} is optimal for certain families of states, up to constant factors. Consider the family of pairs $U\rho U^\dag$, $U \sigma U^\dag$ for all $U \in U(d)$, where $\rho$ and $\sigma$ are proportional to projectors onto orthogonal $r$-dimensional subspaces of $\C^d$. Then
\[ \frac{\|\rho-\sigma\|_1}{\|\rho-\sigma\|_2} = \sqrt{\rank(\rho-\sigma)} = \sqrt{2r}, \]
implying that embeddings of this family with constant distortion and failure probability have a lower bound on the target dimension of $\Omega(\sqrt{rd})$, which is achieved by the embedding of Theorem~\ref{thm:trupper}.


\section{Conclusions}

We have shown that in the 2-norm, any constant-distortion embedding of a unitarily invariant set of $d$-dimensional states must have target dimension $\Omega(d)$, in contrast to the classical situation where an exponential reduction can be achieved. In the trace norm, the situation is somewhat better: $d$-dimensional states of rank $r$ can be embedded in $O(\sqrt{rd})$ dimensions with constant distortion, but there is a lower bound of $\Omega(\sqrt{d} \frac{\|\rho-\sigma\|_1}{\|\rho-\sigma\|_2})$ dimensions on any constant distortion embedding that succeeds for the pairs of states $U \rho U^\dag$ and $U \sigma U^\dag$, for all unitary $U$.

Although the trace distance is often the most physically relevant distance measure to consider, we also argued that for certain tasks, the 2-norm distance is in fact the relevant distance measure between states. This occurs when the basis in which the states were prepared is unknown or the measurement apparatus does not depend on the states to be distinguished.

The alert reader will have noticed that, in the case where one is interested in embedding a unitarily invariant set of states, the embedding might as well start by performing a random unitary. Furthermore, as any quantum channel can be represented as an isometry into a larger space followed by tracing out a subsystem, this makes any embedding seem somewhat similar to the embedding used in Theorem \ref{thm:trupper}. But note that the latter embedding is subtly different, as it can be seen as performing a fixed isometry followed by a random unitary, rather than vice versa. Further analysis of this embedding might allow the gap between the upper and lower bounds in the trace norm to be closed.

Another open question is whether bounds could be obtained on the possible dimensionality reduction when multiple copies of the input state are available. For example, if a very large number of copies are allowed, tomography can be performed, the input state can be approximately determined, and the JL Lemma applied. Presumably, even for a lower number of copies, stronger dimensionality reduction is possible than in the single-copy case. 
One could also ask whether stronger dimensionality reduction can be achieved by allowing some additional classical information; for some results in this direction, see \cite{fawzi10}.

\section*{Acknowledgements}

AWH was supported by the EC grant QESSENCE and the DARPA-MTO QuEST
program through a grant from AFOSR. AM was supported by an EPSRC Postdoctoral Research Fellowship. AJS was supported by the Royal Society.


\appendix


\section{Lemmas relating to 2-norm embeddings}
\label{sec:2normlemmas}

We now prove the subsidiary lemmas required for the proof of Lemma \ref{lem:avg}.

\newcounter{skip}\setcounter{skip}{\value{thm}}
\setcounter{thm}{\value{lems}}

\begin{lem}
Let $\mathcal{E}:\mathcal{B}(\C^d) \rightarrow \mathcal{B}(\C^e)$ be a quantum channel (CPTP map). Then
\[ \tr[F_e\,\mathcal{E}^{\otimes 2}(F_d)] \le de. \]
\end{lem}

\begin{proof}

Assume that $\mathcal{E}$ has the Kraus (operator-sum) decomposition
\[ \mathcal{E}(\rho) = \sum_i A_i \rho A_i^\dag \]
for some $e \times d$ matrices $A_i$ such that $\sum_i A_i^\dag A_i = I_d$, and $\tr[A_i^\dag A_j] = 0$ if $i \neq j$. (Note that such a representation does indeed exist, from the unitary freedom in the Kraus decomposition \cite[Theorem 8.2]{nielsen00}.) Then write
\beas
\tr[F_e\,\mathcal{E}^{\otimes 2}(F_d)] &=& \tr \sum_{i,j} F_e (A_i \otimes A_j) F_d (A_i^\dag \otimes A_j^\dag) = \sum_{i,j} \tr[(A_j \otimes A_i)(A_i^\dag \otimes A_j^\dag)]\\
&=& \sum_{i,j} \tr[A_j A_i^\dag] \tr [A_i A_j^\dag] = \sum_i (\tr[A_i^\dag A_i])^2\\
&\le& \left(\sum_i \tr[A_i^\dag A_i]\right) \max_j \tr[A_j^\dag A_j] \le de.
\eeas
The fourth equality uses the orthogonality of the $A_i$ and cyclicity of the trace, and the final inequality uses the facts that $\sum_i A_i^\dag A_i = I_d$ and $\tr[A_i^\dag A_i] \le \|A_i^\dag A_i\|_\infty \rank(A_i^\dag A_i) \le e$.
\end{proof}

\begin{lem}
Let $\rho$ and $\sigma$ be $d$-dimensional quantum states. Then
\[ \int U^{\otimes 2} (\rho - \sigma)^{\otimes 2} (U^{\dag})^{\otimes 2} dU = \frac{\|\rho-\sigma\|_2^2}{d^2-1} \left( F_d - \frac{I_{d^2}}{d} \right). \]
\end{lem}

\begin{proof}
For brevity, set $\tau :=\int U^{\otimes 2} (\rho - \sigma)^{\otimes 2} (U^{\dag})^{\otimes 2}dU$. Because of the averaging (``twirling'') over the unitary group, $\tau$ must be a linear combination of the identity and swap operators on the space of two $d$-dimensional systems \cite[Theorem 4.2.10]{goodman09}. To evaluate this, we write $\tau = \alpha I_{d^2} + \beta F_d$ and calculate
\[ \tr[\tau] = 0,\; \tr[F_d\,\tau] = \tr[(\rho-\sigma)^2], \]
implying that
\[ \alpha d^2 + \beta d = 0, \; \alpha d + \beta d^2 = \tr[(\rho-\sigma)^2]. \]
Solving for $\alpha$ and $\beta$ gives the claimed result.
\end{proof}


\section{Proof of Theorem \ref{thm:2norm}}
\label{sec:isoproof}

We follow the strategy of Matthews, Wehner and Winter \cite{matthews09a} to prove Theorem \ref{thm:2norm}. We will use two subsidiary results, which are formalised as separate lemmas.

\setcounter{thm}{\value{skip}}

\begin{lem}
\label{lem:2ndmoment}
Let $\rho$, $\sigma$ be $d$-dimensional quantum states. Then
\[ \int d\psi \bracket{\psi}{(\rho-\sigma)}{\psi}^2 = \frac{\tr[(\rho-\sigma)^2]}{d(d+1)}. \]
\end{lem}

\begin{proof}
We use the tensor product trick:
\beas
\int d\psi \bracket{\psi}{(\rho-\sigma)}{\psi}^2 = \int d\psi \tr [(\rho-\sigma)^{\otimes 2} \proj{\psi}^{\otimes 2}]
= \tr \left[(\rho-\sigma)^{\otimes 2} \frac{I_{d^2} + F_d}{d(d+1)} \right]
= \frac{\tr[(\rho-\sigma)^2]}{d(d+1)},
\eeas
noting that $\rho - \sigma$ is traceless and that $\int d\psi(\proj{\psi}^{\otimes 2})$ is proportional to the projector onto the symmetric subspace of two $d$-dimensional systems.
\end{proof}

\begin{lem}
\label{lem:4thmoment}
Let $\rho$, $\sigma$ be $d$-dimensional quantum states. Then
\[ \int d\psi \bracket{\psi}{(\rho-\sigma)}{\psi}^4 \le \frac{9 \tr[(\rho-\sigma)^2]^2}{d(d+1)(d+2)(d+3)}. \]
\end{lem}

\begin{proof}
This is the same technique as the previous lemma, but is a little more involved. Writing
\[ \int d\psi \bracket{\psi}{(\rho-\sigma)}{\psi}^4 = \tr \left[(\rho-\sigma)^{\otimes 4} \int d\psi (\proj{\psi}^{\otimes 4}) \right], \]
we note that $\int d\psi(\proj{\psi}^{\otimes 4})$ is proportional to the projector onto the symmetric subspace of four $d$-dimensional systems, which we write as
\[ P_{sym} = \frac{1}{4!} \sum_{\sigma \in S_4} P_{\sigma}, \]
where $S_4$ is the symmetric group of order 4 and $P_{\sigma}$ is the operator that permutes the 4 systems according to the permutation $\sigma$. Let $\operatorname{Cyc}(\sigma)$ denote the sequence of cycle lengths in $\sigma$ (e.g.\ $\operatorname{Cyc}((12)(3)) = (2,1)$). Then, for any $d$-dimensional operator $X$, it holds that
\[ \tr[X^{\otimes 4} P_{\sigma}] = \prod_{c \in \operatorname{Cyc}(\sigma)} \tr [X^c], \]
which can be shown diagrammatically or by explicitly writing out the $P_{\sigma}$ matrix. In particular, $\tr P_{\sigma} = d^{|\operatorname{Cyc}(\sigma)|}$. Permutations of 4 elements break down into 5 conjugacy classes, as follows: there is 1 of the form $(1)(2)(3)(4)$; 6 of the form $(12)(3)(4)$; 3 of the form $(12)(34)$; 8 of the form $(123)(4)$; and 6 of the form $(1234)$.

Thus
\[ \tr P_{sym} = \frac{1}{4!} (d^4 + 6d^3 + 11d^2 + 6d) = \frac{d(d+1)(d+2)(d+3)}{4!}, \]
implying that
\[ \int d\psi(\proj{\psi}^{\otimes 4}) = \frac{1}{d(d+1)(d+2)(d+3)} \sum_{\sigma \in S_4} P_{\sigma}. \]
We can now calculate
\[ \tr \left[(\rho-\sigma)^{\otimes 4} \int d\psi (\proj{\psi}^{\otimes 4}) \right] = \frac{1}{d(d+1)(d+2)(d+3)} \left( 3\tr[(\rho-\sigma)^2]^2 +6 \tr[(\rho-\sigma)^4] \right), \]
where we use the fact that $\rho-\sigma$ is traceless to ignore all terms corresponding to permutations with fixed points. The upper bound claimed in the statement of the theorem follows by simply noting that $\tr[(\rho-\sigma)^4] \le \tr[(\rho-\sigma)^2]^2$.
\end{proof}

We are finally ready to prove Theorem \ref{thm:2norm}, which we restate for convenience.

\setcounter{skip}{\value{thm}}
\setcounter{thm}{\value{thm2}}

\begin{thm}
Let $\rho$, $\sigma$ be $d$-dimensional quantum states. Then
\[ \frac{1}{3} \|\rho-\sigma\|_2 \le d \int d\psi |\bracket{\psi}{(\rho-\sigma)}{\psi}| \le \|\rho-\sigma\|_2. \]
\end{thm}

\begin{proof}
The upper bound is straightforward:
\[ d \int d\psi |\bracket{\psi}{(\rho-\sigma)}{\psi}| \le d \left( \int d\psi \bracket{\psi}{(\rho-\sigma)}{\psi}^2 \right)^{1/2} = d \left( \frac{\tr[(\rho-\sigma)^2]}{d(d+1)} \right)^{1/2} \le \|\rho-\sigma\|_2, \]
where the first inequality is Jensen's inequality, and the equality is Lemma \ref{lem:2ndmoment}. For the lower bound, we use the fourth moment method of Berger \cite{berger97} (which is just H\" older's inequality in disguise). This states that, for any real-valued random variable $X$,
\[ \E[|X|] \ge \frac{\E[X^2]^{3/2}}{\E[X^4]^{1/2}}. \]
Applying this inequality gives
\[ d \int d\psi |\bracket{\psi}{(\rho-\sigma)}{\psi}| \ge d \frac{\left( \int d\psi \bracket{\psi}{(\rho-\sigma)}{\psi}^2 \right)^{3/2} }{\left( \int d\psi \bracket{\psi}{(\rho-\sigma)}{\psi}^4 \right)^{1/2}} \ge d \left(\frac{\tr[(\rho-\sigma)^2]}{d(d+1)}\right)^{3/2} \left( \frac{d(d+1)(d+2)(d+3)}{9 \tr[(\rho-\sigma)^2]^2} \right)^{1/2} \] 
by Lemmas \ref{lem:2ndmoment} and \ref{lem:4thmoment}, which simplifies to
\[ d \int d\psi |\bracket{\psi}{(\rho-\sigma)}{\psi}| \ge \frac{(d+2)^{1/2}(d+3)^{1/2}}{3(d+1)} \|\rho-\sigma\|_2 \ge \frac{1}{3} \|\rho-\sigma\|_2\]
as claimed.
\end{proof}


\section{Proof of Lemma \ref{lem:projsupp}}
\label{sec:projsupp}

We now prove Lemma \ref{lem:projsupp}, which we restate for convenience.

\setcounter{thm}{\value{thm3}}

\begin{lem}
Let $\mathcal{H} = \mathcal{H}_A \otimes \mathcal{H}_B$ be a finite-dimensional Hilbert space decomposed into subsystems $A$ and $B$. For any projector $P$ onto a subspace of $\mathcal{H}$, let $P^\perp=I-P$ be the projector onto the orthogonal subspace, and let $D$ be the projector onto the support of $\tr_B P$. Then, for any $\ket{\psi} \in \mathcal{H}$,
\[ \tr[ (D \otimes I)P^\perp \proj{\psi} P^\perp] \le \tr[(D \otimes I) \proj{\psi}] \tr[P^\perp \proj{\psi}]. \]
\end{lem}

\begin{proof}
The inequality clearly holds if $\tr[P^\perp \proj{\psi}] = 0$, so assuming this is not the case and dividing both sides by $\tr[P^\perp \proj{\psi}]$, the left-hand side is equal to
\[ \frac{\tr[ (D \otimes I)(I-P) \proj{\psi} (I-P)]}{1 - \tr[P \proj{\psi}]}. \]
The key observation which will allow us to simplify this expression is that $(D \otimes I)P = P = P(D \otimes I)$. To see this, note that the support of $P$ is contained within the subspace onto which $D \otimes I$ projects, implying that $D \otimes I$ acts as the identity with respect to $P$. The left-hand side thus simplifies to
\[ \frac{\tr[ (D \otimes I)\proj{\psi}] - \tr[ P \proj{\psi}]}{1 - \tr[P \proj{\psi}]} \le \frac{\tr[ (D \otimes I)\proj{\psi}](1 - \tr[ P \proj{\psi}])}{1- \tr[P \proj{\psi}]} = \tr[ (D \otimes I)\proj{\psi}] \]
as claimed.
\end{proof}



\end{document}